\documentclass[reqno, 12pt]{amsart}

\usepackage{amsmath,amsthm,bbm,nicefrac,amssymb,braket}
\usepackage[protrusion=true,babel=true]{microtype}
\usepackage[english]{babel}
\usepackage[widespace]{fourier}
\usepackage[backrefs]{amsrefs}
\usepackage[margin=1in]{geometry}
\usepackage[onehalfspacing]{setspace}
\usepackage[pdfusetitle, pagebackref]{hyperref}

\usepackage[dvipsnames]{xcolor}

\numberwithin{equation}{section}

\usepackage[nameinlink, noabbrev]{cleveref}
\expandafter\def\csname ver@etex.sty\endcsname{3000/12/31}

\usepackage{autonum}

\usepackage{tikz}

\usetikzlibrary{quantikz2}
\usetikzlibrary{shapes.geometric, arrows}
\tikzstyle{process} = [ellipse, minimum width=3em, minimum height=2em, text centered, draw=blue, fill=gray!10]
\tikzstyle{arrow} = [thick,->,>=stealth]

\usepackage{subcaption}

\definecolor{accessblue}{cmyk}{1,.3677,0,.3922}
\definecolor{greycolor}{cmyk}{0,0,0,.8}

\let\originalleft\left
\let\originalright\right
\renewcommand{\left}{\mathopen{}\mathclose\bgroup\originalleft}
\renewcommand{\right}{\aftergroup\egroup\originalright}

\def\({\mathopen{}\left(}
\def\){\right)\mathclose{}}

\makeatletter
\renewcommand*{\eqref}[1]{\hyperref[{#1}]{\textup{\tagform@{\ref*{#1}}}}}
\makeatother

\newcommand*{\eqdef}{\mathrel{\vcenter{\baselineskip0.5ex \lineskiplimit0pt\hbox{.}\hbox{.}}}=}
\newcommand*{\defeq}{=\mathrel{\vcenter{\baselineskip0.5ex \lineskiplimit0pt\hbox{.}\hbox{.}}}}

\newtheorem{theorem}{Theorem}[section]
\newtheorem{lemma}[theorem]{Lemma}
\newtheorem{conjecture}[theorem]{Conjecture}
\newtheorem{remark}[theorem]{Remark}

\crefname{theorem}{Theorem}{Theorems}
\creflabelformat{theorem}{#2{#1}#3}
\crefname{lemma}{Lemma}{Lemmas}
\creflabelformat{lemma}{#2{#1}#3}
\Crefname{conjecture}{Conjecture}{Conjectures}
\creflabelformat{conjecture}{#2{#1}#3}
\crefname{ineq}{inequality}{inequalities}
\creflabelformat{ineq}{#2{\upshape(#1)}#3}
\crefname{diag}{diagram}{diagrams}
\creflabelformat{diag}{#2{\upshape(#1)}#3}
\crefname{remark}{Remark}{Remarks}
\creflabelformat{remark}{#2{#1}#3}
\Crefname{figure}{Figure}{Figures}
\creflabelformat{figure}{#2{#1}#3}

\def\id{\mathbbm{1}}

\def\F{\mathbb{F}}

\def\rl{\mathbb{R}}
\def\P{\mathbb{P}}
\def\Z{\mathbb{Z}}

\def\rd{\mathrm{d}}
\def\cF{\mathcal{F}}
\def\cM{\mathcal{M}}
\def\cP{\mathcal{P}}
\def\QFT{\mathrm{QFT}}

\title{Fixed-point Grover Adaptive Search for Quadratic Binary Optimization Problems}
\date{\today}

\author{\'Akos Nagy}
\address[\'Akos Nagy]{BEIT Canada, Toronto, Ontario}
\email{\href{mailto:contact@akosnagy.com}{contact@akosnagy.com}}
\urladdr{\href{https://akosnagy.com/}{akosnagy.com}}

\author{Jaime Park}
\address[Jaime Park]{Vanderbilt University, Nashville, Tennessee}
\email{\href{mailto:jaime.s.park@vanderbilt.edu}{jaime.s.park@vanderbilt.edu}}

\author{Cindy Zhang}
\address[Cindy Zhang]{}
\email{\href{mailto:xindizhang.phys@gmail.com}{xindizhang.phys@gmail.com}}

\author{Atithi Acharya}
\address[Atithi Acharya]{Rutgers, The State University of New Jersey}
\email{\href{mailto:atithi.8@rutgers.edu}{atithi.8@rutgers.edu}}

\author{Alex Khan}
\address[Alex Khan]{University of Maryland QLab, College Park, Maryland}
\email{\href{mailto:askhan@umd.edu}{askhan@umd.edu}}

\hypersetup{
    unicode         = true,
    pdffitwindow    = true,
    pdftoolbar      = false,
    pdfmenubar      = false,
    pdfstartview    = {FitH},
    hypertexnames   = true,
    colorlinks      = true,
    linkcolor       = black,
    citecolor       = black,
    filecolor       = black,
    urlcolor        = blue
}

\calclayout
\begin{document}

\begin{abstract}
    We study a Grover-type method for Quadratic Unconstrained Binary Optimization (QUBO) problems. For an $n$-dimensional QUBO problem with $m$ nonzero terms, we construct a marker oracle for such problems with a tuneable parameter, $\Lambda \in \left[ 1, m \right] \cap \Z$. At $d \in \Z_+$ precision, the oracle uses $O \( n + \Lambda d \)$ qubits, has total depth of $O \( \tfrac{m}{\Lambda} \log_2 (n) + \log_2 \( d \) \)$, and non-Clifford depth of $O \( \tfrac{m}{\Lambda} \)$. Moreover, each qubit required to be connected to at most $O \( \log_2 \( \Lambda + d \) \)$ other qubits. In the case of a maximum graph cuts, as $d = 2 \left\lceil \log_2 (n) \right\rceil$ always suffices, the depth of the marker oracle can be made as shallow as $O \( \log_2 (n) \)$. For all values of $\Lambda$, the non-Clifford gate count of these oracles is strictly lower (at least by a factor of $\sim 2$) than previous constructions.
    
    Furthermore, we introduce a novel \textit{Fixed-point Grover Adaptive Search for QUBO Problems}, using our oracle design and a hybrid Fixed-point Grover Search, motivated by the works of Boyer et al. and Li et al. This method has better performance guarantees than previous Grover Adaptive Search methods. Some of our results are novel and useful for any method based on Fixed-point Grover Search. Finally, we give a heuristic argument that, with high probability and in $O \( \tfrac{\log_2 (n)}{\sqrt{\epsilon}} \)$ time, this adaptive method finds a configuration that is among the best $\epsilon 2^n$ ones.
\end{abstract}

\maketitle

\section{Introduction}

In \cite{grover_quantum_97}, Grover introduced a quantum search algorithm that can be viewed as an unstructured search algorithm whose query complexity scales with the inverse square root of the ratio of the number of marked items to all items. One drawback of Grover's algorithm is the ``souffl\'e problem'', that is, the fact that one needs to know the exact ratio of marked items to compute the correct number of queries prescribed by the algorithm. Not enough or too many queries can ``under/overcook'' the quantum state. This problem was partially remedied by the Fixed-point Grover Search (FPGS) of Yoder et al. \cite{yoder_fixed_2014}. The FPGS is similar to Grover's algorithm in many ways: both are quantum amplitude amplification methods, both have the same query complexity, and they can typically be implemented using the same oracles. Geometrically, the difference is that while Grover's algorithm uses reflections in the real plane spanned by the initial and target states, FPGS uses phase gates in the complexification of that subspace. The benefit of using FPGS is that one no longer needs to know the exact number of marked states; it is enough to only have a lower bound for that number. An important caveat is that the lower this bound is, the more queries are required, but one cannot overcook the quantum state. The query complexity of FPGS scales with the inverse square root of \textit{lower bound used for the ratio of marked states}. Thus, if this lower bound is close to the exact ratio, this latter fact can superficially be viewed as retaining the ``quadratic speedup''. However, that is somewhat misleading, as finding lower bounds might still be hard. In order to circumvent this issue, many adaptive methods have been designed that retain the quadratic speedup; the most important ones for our work are \cites{boyer_tight_1998,li_quantum_2019}. In particular, we refine the ideas of Li et al. \cite{li_quantum_2019} to get an adaptive search method that, to our knowledge, has the best current performance guarantee, in terms of total expected oracle calls and can be applied to arbitrary boolean quantum search.

A QUBO problem consists of a (real-valued) quadratic polynomial on $n$ boolean variables, and thus it can be described, uniquely, up to an overall additive constant, by an $n$-by-$n$ upper-triangular real matrix. Such problems provide some of the most interesting NP-Complete problems, such as cluster analysis, maximum graph cuts, and the Ising model. Algorithms to find approximate maxima of QUBOs have been extensively studied for a long time via classical methods. Recently, classical--quantum hybrid methods have been proposed, where the parameters of the quantum algorithm are classically optimized. The two main types of such algorithms are of QAOA-type \cites{farhi_quantum_2014,szegedy_qaoa_2019,marsh_combinatorial_2020,golden_threshold_2021} and Grover-type \cites{gilliam_grover_2021,giuffrida_solving_2021,sano_accelerating_2023}. Compared to QAOA, Grover-type algorithms have certain set advantages and drawbacks. The promise of QAOA is that, using an easily implementable, low-depth circuit, one can prepare a quantum state whose dominant components in the computational basis correspond to high value configurations of the QUBO problem. On the other hand, Grover-type algorithms tend to have more complex circuits and the amplification is only sensitive to whether or not a configuration is above a given threshold which is a hyperparameter of the method. However, what makes Grover-type methods still appealing is the fact that they usually require much less classical optimizations.

This paper explores a Grover-type method for QUBO problems. In \cite{gilliam_grover_2021}, Gilliam et al. proposed a marker oracle design for QUBO problems: given an instance of an integer QUBO problem, say, by $n$-by-$n$ upper-triangular integer matrix, $Q$, and a threshold $y \in \Z$, their oracle marks states $\ket{x}$ such that $x^T Q x > y$. This is precisely the ingredient needed to run Grover-type search algorithms to find configurations with values above the threshold. Furthermore, in the same paper, Gilliam et al. propose a \textit{Grover Adaptive Search} for QUBO problems, based on the adaptive Grover search of Boyer et al. \cite{boyer_tight_1998}. We explain their construction for the marker oracle in \Cref{sec:qubo_encoder} and their Grover Adaptive Search in \Cref{sec:adaptive}. This paper was inspired by the results of Gilliam et al. and our key results can be viewed as improvements on both their oracle design and adaptive method. Most notably, our method has improved depth and total gate counts for both for Clifford and non-Clifford gates. Furthermore, we provide a variant of the oracle with a tuneable number of ancillas, which uses the same (amount and type of) non-Clifford gates and has reduced the total circuit depth, while having only an $O \( 1 \)$-fold increase in the Clifford gate count. In the shallowest case the depth is $O \( \log_2 \( d \) \)$---in particular, independent of $n$---and uses at most $\tfrac{n^2 + n + 2}{2} d$ ancillas. While we focus on unconstrained problems, constraints can readily be handled, as in \cite{gilliam_grover_2021}*{Figure~7}. Finally, our methods can be combined with other, instance-specific circuit optimization techniques; for example \cites{sano_accelerating_2023}.

\smallskip

\subsection*{Organization of the paper:} In \Cref{sec:qdicts_and_qubos}, we introduce QUBO problems and present our design for QUBO quantum dictionary encoders. In \Cref{sec:grover_for_qubo}, we introduce the Fixed-point Grover Search and, using our oracle, apply it to QUBO problems. \Cref{sec:adaptive} defines and studies the adaptive version of the Fixed-point Grover Search. Finally, in \Cref{sec:comparison}, we test and compare our method to the Grover Adaptic Search for QUBO of Gilliam et al. in \cite{gilliam_grover_2021}.

\bigskip

\section{Quantum Dictionaries and QUBO problems}
\label{sec:qdicts_and_qubos}

Given a function $f : \F_2^n \rightarrow \rl$, the associated (Unconstrained) Binary Optimization problem is the task of finding an element $x \in \F_2^n$ such that $f (x)$ is maximal. Many interesting Binary Optimization problems are quadratic, such as finding maximum graph cuts or the Max $2$-SAT problems, and thus, most of the contemporary research centers around Quadratic Unconstrained Binary Optimization (QUBO) problems.

The first main contribution of the paper is an oracle design for QUBO problems. More concretely, we define and construct (projective) encoder oracles for quadratic functions $f : \F_2^n \rightarrow \left[ - 2^{d - 1}, 2^{d - 1} \right) \cap \Z \cong \F_2^d$ below. These oracles have applications in, for example, Grover-type algorithms and threshold QAOA \cites{golden_threshold_2021,bridi_analytical_2024}. While designs for such oracles have already existed, cf. \cite{gilliam_grover_2021}, ours has better circuit depth and size, both for Clifford and non-Clifford gates. Furthermore, we construct a variant with quadratically many ancillas, but also with logarithmic depth. A generalization of these oracles is given in \cite{nagy_novel_2024}, which can be applied to polynomials of arbitrary degrees.

\smallskip

In this paper we use two common labelling for the computational basis. First, given a bit string $x = \( x_0, x_1, \ldots, x_{n - 1} \) \in \F_2^n$, we write $\ket{x} \eqdef \ket{x_0} \ket{x_1} \cdots \ket{x_{n - 1}}$. Similarly, if $\bar{y} \in \Z$ and $d \in \Z_+$, then let $y = \( y_0, y_1, \ldots, y_{d - 1} \) \in \F_2^d$ be such that $\bar{y} \equiv \sum_{j = 0}^{d - 1} y_j 2^{d - 1 - j}$ and $\ket{\bar{y}}_d \eqdef \ket{y}$, that is, the state is the $\bar{y}^{\mathrm{th}} \ \( \textnormal{mod } 2^d \)$ computational basis element in the usual ordering.

With the above notation, a \textit{quantum dictionary encoder}, $U_{f, d}$, corresponding to a function, $f : \F_2^n \rightarrow \Z$, is any oracle on $n + d$ qubits defined via
\begin{equation}
    U_{f, d} \ket{x} \ket{y}_d \eqdef \ket{x} \ket{y + f (x)}_d.
\end{equation}
More generally, $U_{f, d}$ is a \textit{projective} quantum dictionary encoder, if
\begin{equation}
    U_{f, d} \ket{x} \ket{y}_d \eqdef e^{i \alpha_{f, x, d}} \ket{x} \ket{y + f (x)}_d, \label{eq:projective_encoder}
\end{equation}
for some $\alpha_{f, x, y} \in \rl$.

\begin{remark}
    In many quantum algorithms, including ours, quantum dictionary encoders appear in pairs in the form
    \begin{equation}
        \begin{quantikz}
            \cdots \hspace{5mm}  & \gate{U_{f, d}}   & \gate{D}  & \gate{U_{f, d}^\dagger}   & \hspace{5mm} \cdots
        \end{quantikz}
    \end{equation}
    where $D$ is diagonal in the computational basis. Thus, using a projective encoder yields the same outcome, since
    \begin{align}
        U_{f, d}^\dagger D U_{f, d} \ket{x} \ket{y}_d   &= e^{i \alpha} U_{f, d}^\dagger D \ket{x} \ket{y + f (x)}_d \\
            &= \eta U_{f, d}^\dagger \( e^{i \alpha} \ket{x} \ket{y + f(x)}_d \) \\
            &= \eta \ket{x} \ket{y}_d,
    \end{align}
    where $\eta$ is the appropriate eigenvalue of $D$.
\end{remark}

\bigskip

\subsection{QUBO encoder}
\label{sec:qubo_encoder}

Let $\QFT_d$ denote the $d$-qubit Quantum Fourier Transform (QFT). In order to construct $U_{f, d}$ for a quadratic $f$, we use Draper's QFT based shift operator \cite{draper_addition_2000}, which we summarize in the following theorem:

\begin{theorem}
    Let $\tilde{U}_{f, d}$ be defined as
    \begin{equation}
        \tilde{U}_{f, d} \ket{x} \ket{y}_d = \exp \( \tfrac{2 \pi i}{2^d} f \( x \) y \) \ket{x} \ket{y}_d. \label{eq:f_phaser}
    \end{equation}
    Then
    \begin{equation}
        U_{f, d} = \( \id \otimes \QFT_d^\dagger \) \circ \tilde{U}_{f, d} \circ \( \id \otimes \QFT_d \). \label{eq:quantum_adder}
    \end{equation}
    is a quantum dictionary encoder for $f$.
\end{theorem}

The above theorem allows the computation of $f(x)$ to be done in the ``phase-space'' and reduces the construction of $U_{f, d}$ to the construction of $\tilde{U}_{f, d}$, which we do below.

\smallskip

Before we present the design of $\tilde{U}_{f, d}$, let us make a few definitions. Let now $f : \F_2^n \rightarrow \left[ - 2^{d - 1}, 2^{d - 1} \) \cap \Z$ be quadratic, and write
\begin{equation}
    f(x) = f(0) + \sum\limits_{j, k = 1}^n Q_{jk} x_j x_k,
\end{equation}
where $\( Q_{jk} \)_{j, k \in [n]}$ is a real, symmetric matrix. Let $\oplus$ be the addition modulo $2$ (that is, logical XOR). Then we have
\begin{equation}
    x_j x_k = \tfrac{1}{2} \( x_j + x_k - x_j \oplus x_k \),
\end{equation}
and thus, if we set $q_j \eqdef \sum_k Q_{jk}$, then
\begin{equation}
    f(x) = f(0) + \sum\limits_{j = 0}^{n - 1} q_j x_j - \sum\limits_{j = 0}^{n - 2} \sum\limits_{k = j + 1}^{n - 1} Q_{jk} x_j \oplus x_k. \label{eq:f_rewrite}
\end{equation}
For $k \in \rl$, let us define the $d$-qubit operators
\begin{equation}
    \cP_d \( k \) \eqdef \otimes_{j = 0}^{d - 1} \begin{pmatrix} 1 & 0 \\ 0 & e^{\scriptscriptstyle \frac{\pi i}{2^j} k} \end{pmatrix},
\end{equation}
or, equivalently
\begin{equation}
    \cP_d \( k \) \ket{y}_d \eqdef e^{\scriptscriptstyle \frac{2 \pi i}{2^d} k y} \ket{y}_d.
\end{equation}
Let $\tilde{U}_\emptyset \eqdef \id \otimes \cP_d \( q_\emptyset \)$, where
\begin{equation}
    q_\emptyset \eqdef f(0) + \tfrac{1}{4} \mathrm{tr} \( Q \) + \tfrac{1}{4}  \mathrm{sum} \( Q \),
\end{equation}
and thus
\begin{equation}
    \tilde{U}_\emptyset \ket{x} \ket{y}_d = e^{i \alpha_{\emptyset, y}} \ket{x} \ket{y}_d,
\end{equation}
where
\begin{equation}
    \alpha_{j, y} = \tfrac{2 \pi}{2^d} q_\emptyset y, \label{eq:alpha_nully}
\end{equation}
Let the gates $\tilde{U}_j$ be
\begin{equation}
    \begin{quantikz}
        \lstick{$\ket{x_j}$}    & \ctrl{1}              & \qw                                   & \ctrl{1}              & \rstick{$\ket{x_j}$} \\
        \lstick{$\ket{y}_d$}    & \gate{X^{\otimes d}}  & \gate{\cP_d \( - \tfrac{q_j}{2} \)}    & \gate{X^{\otimes d}}  & \rstick{$e^{i \alpha_{j, y}} \ket{y}_d$}
    \end{quantikz}
\end{equation}
where
\begin{equation}
    \alpha_{j, y} = \tfrac{2 \pi}{2^d} q_j \( x_j y - \tfrac{x_j}{2} \( 2^d - 1 \) - \tfrac{y}{2} \), \label{eq:alpha_jy}
\end{equation}
and the gates $\tilde{U}_{jk}$ be
\begin{equation}
    \begin{quantikz}
        \lstick{$\ket{x_j}$}    & \ctrl{1}    & \qw                     & \qw                                   & \qw                   & \ctrl{1}  & \rstick{$\ket{x_j}$} \\
        \lstick{$\ket{x_k}$}    & \targ{}     & \ctrl{1}                & \qw                                   & \ctrl{1}              & \targ{}   & \rstick{$\ket{x_k}$} \\
        \lstick{$\ket{y}_d$}    & \qw         & \gate{X^{\otimes d}}    & \gate{\cP_d \( \tfrac{Q_{jk}}{2} \)}   & \gate{X^{\otimes d}}  & \qw   & \rstick{$e^{i \alpha_{j, k, y}} \ket{y}_d$}
    \end{quantikz}
\end{equation}
where
\begin{equation}
    \begin{aligned}
        \alpha_{j, k, y}    &= \tfrac{2 \pi}{2^d} Q_{jk} \big( - \( x_j \oplus x_k \) y \\
                            & \quad + \tfrac{x_j \oplus x_k}{2} \( 2^d - 1 \) + \tfrac{y}{2} \big), \label{eq:alpha_jky}
    \end{aligned}
\end{equation}
Note that the \textit{fanout}, or \textit{multi-target CNOT} gates
\begin{equation}
    \begin{quantikz}
        \lstick{$\ket{\mathrm{control}}$}   & \ctrl{1}              & \qw \rstick{$\ket{\mathrm{control}}$} \\
        \lstick{$\ket{y}_d$}                & \gate{X^{\otimes d}}  & \qw \rstick{$\left\{ \begin{array}{ll} \ket{y}_d, & \mbox{if } \mathrm{control} = 0, \\ \ket{ - y - 1}_d, & \mbox{otherwise.} \end{array} \right.$}
    \end{quantikz}
\end{equation}
can be implemented using $\approx 2 d$ many CNOT gates with $\approx 2 \log_2 \( d \)$ depth; cf. \cite{low_trading_2024}*{Appendix~B.1}. With these definitions, let
\begin{equation}
    \tilde{U}_{f, d} \eqdef \tilde{U}_\emptyset \circ \prod_{j = 0}^{n - 1} \tilde{U}_i \circ \prod_{j = 0}^{n - 2} \prod_{k = j + 1}^{n - 1} \tilde{U}_{jk}.
\end{equation}

\begin{theorem}
    \label{theorem:qubo_encoder_no_ancillas}
    For each $x \in \F_2^n$ and $y \in \F_2^d$, we have that
    \begin{equation}
        \tilde{U}_{f, d} \ket{x} \ket{y}_d = \exp \( \tfrac{2 \pi i}{2^d} f \( x \) y + i \alpha_{f, x, y} \) \ket{x} \ket{y}_d,
    \end{equation}
    with $\alpha_{f, x, y} \eqdef \tfrac{\pi \( 1 - 2^d \)}{2^d} \( f \( x \) - f(0) \)$. Thus, using \cref{eq:quantum_adder,eq:f_rewrite} and the above equation we get that
    \begin{equation}
        U_{f, d} \eqdef \( \id \otimes \QFT_d^\dagger \) \circ \( \tilde{U}_{f, d} \) \circ \( \id \otimes \QFT_d \). \label{eq:proj_enc}
    \end{equation}
    is a projective quantum dictionary encoder for $f$, with the same \textit{garbage phases} $\alpha_{f, x, y}$.    
\end{theorem}

\begin{proof}
    The verifications of \cref{eq:alpha_nully,eq:alpha_jy,eq:alpha_jky} is straightforward and they yield
    \begin{equation}
        \tilde{U}_{f, d} \ket{x} \ket{y}_d = \exp \( \tfrac{2 \pi i}{2^d} \varphi_{x, y} \) \ket{x} \ket{y}_d
    \end{equation}
    where
    \begin{align}
        \varphi_{x, y}  &\eqdef q_\emptyset y + \sum\limits_{j = 0}^{n - 1} q_j x_j y - \sum\limits_{j = 0}^{n - 2} \sum\limits_{k = j + 1}^{n - 1} Q_{jk} x_j \oplus x_k y \\
                        & \quad - \tfrac{2^d - 1}{2} \( \sum\limits_{j = 0}^{n - 1} q_j x_j + \sum\limits_{j = 0}^{n - 2} \sum\limits_{k = j + 1}^{n - 1} Q_{jk} x_j \oplus x_k \) \\
                        & \quad - \tfrac{1}{2} \( \sum\limits_{j = 0}^{n - 1} q_j y - \sum\limits_{j = 0}^{n - 2} \sum\limits_{k = j + 1}^{n - 1} Q_{jk} y \) \\
                        & = \( q_\emptyset + f(x) - f(0) \) y \\
                        & \quad - \tfrac{2^d - 1}{2} \( f(x) - f(0) \) \\
                        & \quad - \tfrac{1}{2} \( \mathrm{sum} \( Q \) - \tfrac{1}{2} \( \mathrm{sum} \( Q \) - \mathrm{tr} \( Q \) \) \) y \\
                        & = f(x) y - \tfrac{2^d - 1}{2} \( f(x) - f(0) \),
    \end{align}
    which concludes the proof.
\end{proof}

\smallskip

\begin{remark}
    We have so far restricted our discussion to integer-valued QUBO problems. However, the definition of $\cP_d \( k \)$ makes sense for any $k \in \rl$, not just for integers. When $k \notin \Z$, and we let $[k] \in \Z$ and $\tilde{k} \in \( 0, 1 \)$ be its integer and fractional parts, respectively, then we get
    \begin{equation}
        \QFT_d^\dagger \circ \cP_d \( k \) \circ \QFT_d \ket{y}_d = \sum\limits_{z = 0}^{2^d - 1} \tfrac{e^{2 \pi \tilde{k} i} - 1}{e^{\scriptscriptstyle {2 \pi i}{2^d} \( y - z + k \)} - 1} \ket{z}_d.
    \end{equation}
    This yields a Fej\'er type distribution on the bit strings with the probability of measuring either $y + [k]$ or $y + [k] + 1$ being at least $\tfrac{8}{\pi^2} \approx 81\%$; cf. \cite{gilliam_grover_2021}*{Appendix~B.2.}.

    Furthermore, since $\F_2^n$ is a finite set, small perturbations of the coefficient matrix do not change the order of the values, meaning that small ``roundings'' of the coefficients yield equivalent optimization problems.

    Thus, the same oracle design can be, in some cases, used for real-valued QUBO problems.
\end{remark}

\medskip

\subsection{Using ancillas to reduce depth:}
\label{sec:ancillas}
Let now $m$ be the number of nonzero terms in \cref{eq:f_rewrite} and let $1 \leqslant \Lambda \leqslant m$ be an integer. We claim, without proof, that a circuit for $U_{f, d}$ can be constructed (with the same garbage phases as in \Cref{theorem:qubo_encoder_no_ancillas}) using $\( \Lambda - 1 \) d$ ancillas, where the $R_Z$-depth in general is $\left\lceil \tfrac{m}{\Lambda} \right\rceil$ (assuming that $\Lambda d$ many $R_Z$ gates can be run in parallel) and the total depth is $O \( \tfrac{m}{\Lambda} \max \( \left\{ 1, \log_2 \( \tfrac{\Lambda}{n} \) \right\} \) + \log_2 \( d \) \)$. We only present the construction in detail in the most extreme case, when $m$ is maximal, that is, when $m = 1 + n + \binom{n}{2} = \tfrac{n^2 + n + 2}{2}$. We also focus on the shallowest case, which is also the case with the most ancillas, that is $\Lambda = m$. Versions for more sparse functions and different values of $\Lambda$ can then easily be constructed and we address this issue after our proof for the $m = \Lambda = \tfrac{n^2 + n + 2}{2}$ case, when the $R_Z$ depth is $1$ and the total depth is $O \( \log_2 (n) + \log_2 \( d \) \)$.

\smallskip

Let us assume now that $m = \Lambda = \tfrac{n^2 + n + 2}{2}$ and start with a quantum circuit on $n + m d$ qubits. Our goal is still the realization of $\tilde{U}_{f, d}$, as in \cref{eq:f_phaser}, on the first $n + d$ qubits. We assume that the last $\( m - 1 \) d$ qubits are all initialized in zero. Let $U_{\mathrm{QUBO}}$ be defined as 
\begin{multline}
    U_{\mathrm{QUBO}} \ket{x} \ket{y}_d \ket{0}_{\( m - 1 \) d} \eqdef \\
    \ket{x} \ket{y}_d \bigotimes\limits_{j = 0}^{n - 1} \bigotimes\limits_{a = 0}^{d - 1} \( \ket{y_a \oplus x_j} \bigotimes\limits_{k = j + 1}^{n - 1} \ket{y_a \oplus x_j \oplus x_k} \).
\end{multline}
Note that $U_{\mathrm{QUBO}}$ is independent of $f$.

\begin{theorem}
    \label{theorem:QUBO_preparator}
    $U_{\mathrm{QUBO}}$ can be constructed using only CNOT gates with depth $O \( \log_2 (n) + \log_2 \( d \) \)$.
\end{theorem}

\begin{proof}
    Note that each $y_a$ needs to be added to $m$ qubits and no two of them need to be added to the same one, thus, using again the fanout gates of \cite{low_trading_2024}*{Appendix~B.1}, this can be achieved using $\sim 2 m d$ CNOT gates with a depth of $\sim 2 \log_2 \( m \) \lesssim 4 \log_2 (n)$. After this, each $x_j$ needs to be added to $n d$ terms, albeit some of them overlap, but we can still achieve this with two parallel sets of fanout gates; first (for each $j$) let $\ket{x_j}$ control a fanout gate that targets the qubits we want to end up in the state either $\ket{y_a \oplus x_j}$ or $\ket{y_a \oplus x_j \oplus x_k}$, where $j < k$ and in the second set let $\ket{x_j}$ control a fanout gate that targets the qubits we want to end up in the state $\ket{y_a \oplus x_j \oplus x_k}$, where $j > k$. In each round the fanout gates have no overlapping controls or targets, thus can be run parallel, and the maximum depth in each is $\approx 2 \log_2 (n)$. Thus the total depth is $ \lesssim \( 4 + 2 + 2 \) \log_2 (n) = 8 \log_2 (n)$.
\end{proof}

\smallskip

Now let
\begin{align}
    \P_f \eqdef & \id \otimes \cP_d \( q_\emptyset \) \bigotimes\limits_{j = 0}^{n - 1} \( \cP_d \( - \tfrac{q_i}{2} \) \bigotimes\limits_{k = j + 1}^{n - 1} \cP_d \( \tfrac{Q_{jk}}{2} \) \).
\end{align}
Note that $\P_f$ contains the exact same number and type of $R_Z$ rotations as the design in \Cref{sec:qubo_encoder}, but now they all run in parallel.

\begin{lemma}
    Let $U_{f, d} \eqdef U_{\mathrm{QUBO}} \circ \P_f \circ U_{\mathrm{QUBO}}$. Then for each $x \in \F_2^n$ and $y \in F_2^d$, we have that
    \begin{equation}
        U_{f, d} \ket{x} \ket{y}_d \ket{0}_{\( m - 1 \) d} = e^{i \alpha_{f, x, y}} \ket{x} \ket{y + f(x)}_d \ket{0}_{\( m - 1 \) d}
    \end{equation}
    where $\alpha_{f, x, y}$ is the same as in \Cref{theorem:qubo_encoder_no_ancillas}.
\end{lemma}

\begin{proof}
    Analogous to the proof of \Cref{theorem:qubo_encoder_no_ancillas}.
\end{proof}

\smallskip

\begin{remark}
    Note that the garbage phases can be also eliminated by applying phase gates with angles $\tfrac{2 \pi}{2^d} q_i$ to the qubits $\ket{x_j}$ and controlled phase gates with angles $- \tfrac{2 \pi}{2^d} Q_{jk}$ to the qubits $\ket{x_j} \ket{x_k}$ (for all relevant values of $j, k$). However, for our current purposes a projective encoder is already sufficient, as explained in \Cref{sec:qdicts_and_qubos}.
\end{remark}

\bigskip

\section{Fixed-point Grover search for QUBO}
\label{sec:grover_for_qubo}

Fixed-point Grover Search (FPGS) \cite{yoder_fixed_2014} is a variant of Grover's search algorithm that retains the original version's query complexity while not suffering from the souffl\'e problem (more on this below). In this section we introduce the algorithm of Yoder et al. and use our oracle design from \Cref{sec:qubo_encoder} to implement FPGS for QUBO problems. Furthermore, in \Cref{sec:parameter_optimization_for_FPGAS}, we study the problem of optimal parameter choice for FPGS and, in \Cref{sec:time}, we analyze the time-complexity of FPGS for QUBO problems using our oracle design.

\smallskip

As opposed to Grover's algorithm, where the only input is the set of \textit{good}/\textit{marked}/\textit{target} configurations, $\cM \subset \F_2^n$, the FPGS algorithm requires an additional one, that can be chosen to be either the target probability or the query complexity. We assume that the target probability, $\P_{\mathrm{target}} \in \( 0, 1 \)$, is given, that is, we want the algorithm to produce a state that, when measured in the computational basis, yields a marked state with probability at least $\P_{\mathrm{target}}$.

The price of this flexibility (and of the elimination of half of the souffl\'e problem) is that one needs to implement not only a pair of oracles, but two families of them, call $S_s \( \alpha \)$ and $S_t \( \beta \)$ (where the subscripts refer to the \textit{start} and \textit{target} states, and $\alpha, \beta$ are real parameters), with the following properties: Let
\begin{align}
    \ket{s} &\eqdef \tfrac{1}{\sqrt{2^n}} \sum\limits_{x \in \F_2^n} \ket{x}, \\
    \ket{t} &\eqdef \tfrac{1}{\sqrt{|\cM|}} \sum\limits_{x \in \cM} \ket{x},
\end{align}
let $\lambda \eqdef \left| \braket{s|t} \right|^2 = \tfrac{|\cM|}{2^n} \in \( 0, 1 \)$, and
\begin{equation}
    \ket{\bar{s}} \eqdef \tfrac{\sqrt{\lambda} \ket{s} - \ket{t}}{\sqrt{1 + \lambda}} \: \& \: \ket{\bar{t}} \eqdef \tfrac{- \ket{s} + \sqrt{\lambda} \ket{t}}{\sqrt{1 + \lambda}}.
\end{equation}
Then both $\left\{ \ket{s}, \ket{\bar{s}} \right\}$ and $\left\{ \ket{t}, \ket{\bar{t}} \right\}$ are orthonormal bases of the same complex plane. Now we need $S_s \( \alpha \)$ and $S_t \( \beta \)$ to be unitaries such that on this plane their actions are 
\begin{align}
    S_s \( \alpha \) \( A \ket{s} + B \ket{\overline{s}} \) &= e^{i \alpha} A \ket{s} + B \ket{\overline{s}} \quad \\
                                                            & \: \& \\
    S_t \( \beta \) \( D \ket{t} + C \ket{\overline{t}} \)  &= e^{i \beta} C \ket{t} + D \ket{\overline{t}}.
\end{align}
Let $G \( \alpha, \beta \) \eqdef S_s \( \alpha \) S_t \( \beta \)$.

Once in possession of such oracles and a target probability $\P_{\mathrm{target}} \in \( 0, 1 \)$, the FPGS of \cite{yoder_fixed_2014} can be summarized as follows: for each $\mu \in \( 0, \lambda \right]$ and $\delta \in \( 0, \sqrt{1 - \P_{\mathrm{target}}} \right]$, there is $l = l \( \mu, \delta \) \in \Z_+$, such that if, for all $j \in \{ 1, 2, \ldots, l \}$, we set
\begin{align}
    \alpha_j    &\eqdef 2 \: \mathrm{arccot} \( \tan \( \tfrac{2 \pi j}{2 l + 1} \) \tanh \( \tfrac{\mathrm{arccosh} \( \nicefrac{1}{\delta} \)}{2 l + 1} \) \), \\
    G_j         &\eqdef S_s \( \alpha_j \) S_t \( \alpha_{l + 1 - j} \),
\end{align}
then
\begin{equation}
    \P_{\mathrm{success}} \eqdef \left| \left\langle t \middle| G_l G_{l - 1} \cdots G_1 \middle| s \right\rangle \right|^2 \geqslant \P_{\mathrm{target}}.
\end{equation}
In fact, $l$ and $\P_{\mathrm{success}}$ can be computed via:
\begin{align}
    l                       &= \min \( \left\{ \: l^\prime \: \middle| \: \sqrt{1 - \lambda} T_{\nicefrac{1}{2 l + 1}} \( \nicefrac{1}{\delta} \) > 1 \: \right\} \), \label{eq:lcrit} \\
    \P_{\mathrm{success}}   &= 1 - \( T_{2 l + 1} \( \sqrt{1 - \lambda} T_{\nicefrac{1}{2 l + 1}} \( \nicefrac{1}{\delta} \) \) \)^2 \delta^2, \label{eq:Psuccess}
\end{align}
where $T_L$ is the Chebyshev polynomial of the first kind. Now
\begin{equation}
    \P_{\mathrm{success}} \geqslant 1 - \delta^2 \geqslant \P_{\mathrm{target}}.
\end{equation}
The proof of the above statements can be found in \cite{yoder_fixed_2014}.

\smallskip

\begin{remark}
    In \cite{grover_fixed_2005} Grover also introduced a ``fixed-point'' version of his original method, however this algorithm did not maintain the quadratic speedup.
\end{remark}

\begin{remark}
    The fixed-point property is crucial for our method, because the number of marked states for QUBO problems is unknown. This rules out the direct application of other known methods such as ``quantum search with certainty'' \cite{toyama_quantum_2013} or  ``nonboolean'' versions of quantum search (for example \cites{benchasattabuse_amplitude_2022,shyamsundar_nonboolean_2023}). This can be somewhat remedied by the Quantum Mean Estimation method of \cite{shyamsundar_nonboolean_2023}, it is unclear if such a combined method would not lose the computational advantages\footnote{Based on private communications with the author of \cite{shyamsundar_nonboolean_2023}.}.
    
    Furthermore, the proof of a fixed-point version of the Nonboolean Quantum Amplitude Amplification (NBQAA) method of \cite{shyamsundar_nonboolean_2023} are currently prepared for publication by the first author, which will open up the possibility of combining the advantages of NBQAA with our adaptive algorithm for binary optimization.
\end{remark}

\smallskip

We implement $S_s \( \alpha \)$ and $S_t \( \beta \)$ in a straightforward manner. If $U_s$ is the state preparation oracle, that is, $\ket{s} = U_s \ket{0}$, and $\mathrm{MCP}_n \( \alpha \)$ is the $(n - 1)$-controlled phase gate on $n$ qubits, then
\begin{equation}
     S_s \( \alpha \) = U_s \circ X^{\otimes n} \circ \mathrm{MCP}_n \( \alpha \) \circ X^{\otimes n} \circ U_s^\dagger, \label{eq:diffuser}
\end{equation}
works. Note that $\mathrm{MCP}_n$ can be implemented with one clean ancillas and two multi-controlled NOT gates via
\begin{equation}
    \begin{quantikz}
        \lstick{$\ket{x}$}  & \qwbundle{n}  & \ctrl{1}  & \qw               & \ctrl{1}  & \rstick{$\ket{x}$} \\
        \lstick{$\ket{0}$}  &               & \targ{}   & \phase{\alpha}    & \targ{}   & \rstick{$e^{i \alpha x_0 x_1 \cdots x_{n - 1}} \ket{0}$}
    \end{quantikz}
\end{equation}
In turn, the $n$-controlled NOT gates can be implemented with a depth of $O \( \log_2 (n)^3 \)$ with no ancillas \cite{claudon_polylogarithmic_2024} or at depth $O \( \log_2 (n) \)$ and a single ancilla \cite{nie_quantum_2024}.

\smallskip

In general, the implementation of $S_t \( \beta \)$ is case specific. In our method in \Cref{sec:adaptive} for QUBO problems, we first specify thresholds, $y \in \Z$, and want to mark states, $\ket{x}$, such that $f(x) > y$, that is, the set of marked states is $\cM_{f, y} \eqdef \{ x \in \F_2^n | y > f(x) \}$. We achive this as follows. First replace $f$ with $f_y (x) \eqdef y - f(x)$ and repeat one of the constructions of \Cref{sec:qubo_encoder}, to get $U_t \eqdef U_{f_y, d}$. Note that $\ket{x}$ is marked exactly when $f_y (x) < 0$, which happens exactly when the first qubit of $\ket{f_y (x)}_d$ is $\ket{1}$ (and otherwise it is $\ket{0}$). Now to to get $S_t \( \beta \)$, we apply both $U_t$ and $U_t^\dagger$ as follows:
\begin{equation}
    \begin{quantikz}
        \lstick{$\ket{x}$} \qw                      & \gate[3]{U_{f_y, d}}    & \qw                               & \gate[3]{U_{f_y, d}^\dagger}    & \qw \rstick{$\left\{ \begin{array}{ll} e^{i \beta} \ket{x}, & \mbox{if } x \in \cM, \\ \phantom{e^{i \beta}} \ket{x}, & \mbox{if } x \notin \cM, \end{array} \right.$} \\
        \lstick{$\ket{0}$} \qw                      &                       & \phase{\beta} &                               & \qw \rstick{$\ket{0}$} \\
        \lstick{$\ket{0}_{d - 1}$}  &                        &                                  &                               & \qw \rstick{$\ket{0}_{d - 1}$} \\
    \end{quantikz}
\end{equation}
Again, since $U_{f_y, d}$ is used in pairs with $U_{f, d}^\dagger$ in a way that the operator between them is diagonal, all garbage phases are uncomputed and thus a projective quantum dictionary encoder is sufficient.

\medskip

\subsection{Parameter Optimization for Fixed-point Grover Adaptive Search}
\label{sec:parameter_optimization_for_FPGAS}

Recall that in FPGS the parameter $\delta$ (or, equivalently, $\P_{\mathrm{target}}$) can be chosen arbitrarily within the interval $\( 0, 1 \)$. In this section we show how to optimize it for our adaptive search.

\smallskip

First of all, just to illustrate our ideas, let us assume that $\lambda$ is known and, thus, we choose $\mu = \lambda$. In order to run FPGS, we need to still choose $\delta$, which then determines $l \defeq l_\delta$ and $\P_{\mathrm{success}} \defeq\P_
\delta$. Since the number of FPGS runs needed to succeed follows the geometric probability distribution, the expected amount of time required to find a solution is proportional to $\tfrac{l_\delta}{P_\delta}$. Since we know that the query complexity scales as $\tfrac{1}{\sqrt{\lambda}}$, let us define $\tau_\delta \eqdef \sqrt{\lambda} \tfrac{l_\delta}{P_\delta}$, and minimize $\tau_\delta$ in the $\lambda \rightarrow 0^+$ limit. This is a nontrivial problem to solve analytically (as the formulas involve nondifferentiable ceiling functions), but we found that this optimization can be carried out numerically, yielding stable approximative values. An instance of this is given in \Cref{figure:delta_min}.

\begin{figure}[h]
    \centering
    \includegraphics[scale=.8]{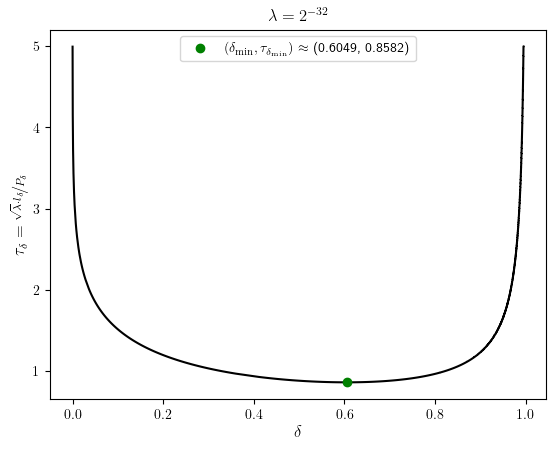}
    \caption{The $\delta$-dependence of the expected time to find a marked state.}
    \label{figure:delta_min}
\end{figure}

We find that $\delta \approx 0.6049$ minimizes the expected time and yields $\tau \approx 0.8582$. As expected, this is higher than $\tfrac{\pi}{4} \approx 0.7854$, which corresponds to Grover's algorithm. In fact, Grover's algorithm can be achieve to $\tau \approx 0.69$, as in pointed out in \cite{zalka_grover_1999}*{Section~1}. Thus, when the ratio of marked states is known, we should not use FPGS, but Grover's algorithm.

However, when $\lambda$ is not known, Grover's algorithm cannot directly be applied. This was addressed in \cite{boyer_tight_1998}: The authors propose a randomized algorithm that finds a marked state, even when $\lambda$ is unkown, and they provide an upper bound $\tau \lesssim \tfrac{9}{2} = 2.25$. Here ``randomized'' means that the number of queries are randomly chosen and $\tau$ is still defined as $\sqrt{\lambda}$ times the expected number of queries, as $\lambda \rightarrow 0^+$. This is the method used in \cite{gilliam_grover_2021}. In \cite{li_quantum_2019}, Li et al. proposed a similar method, using FPGS, that they called \textit{deterministic}, in the sense that the number of queries are chosen deterministicly (growing exponentially using a specified base) and this number surpasses the critical number of queries needed for FPGS (given in \cref{eq:lcrit}) in $O \( \tfrac{1}{\sqrt{\lambda}} \)$ time. However, they only prove $\tau \lesssim 5.643$ for their parameters. Below, we refine their argument to (numerically) achieve $\tau \approx 1.433$. Thus our method yields a deterministic method with the best known worst case complexities. The caveat is that we do not prove our bound analytically, just provide strong, numerical evidence for it down to $\lambda \geqslant 2^{- 40}$.

\begin{remark}
    We assume that queries for FPGS and Grover's algorithm cost the same, which is not always the case: to implement arbitrary phase rotations one often need two marker calls (whereas a phase $\pi$ shift only needs one) and one also need to consume more magic states in a fault-tolerant (Clifford $+$ $T$) implementation. This is not the case for our method however, as in both cases one needs two marker calls to uncompute the ancillas and the non-Clifford overhead of implementing arbitrary phases is small compared to that of the marker oracles.

    Note that often FPGS queries are assumed to cost twice as much as regular Grover queries, in which case our bound is somewhat weaker than that of \cite{boyer_tight_1998}.
\end{remark}

In order to achieve $\tau \approx 1.433$, we mimic the ideas of Li et al. \cite{li_quantum_2019}, however, we use sharper bounds. More precisely, our method can be described as follows: Fix $\delta \in \( 0, 1 \)$, $\alpha > 1$, and let FPGS$_{\delta, l}$ be FPGS with parameter $\delta$ and $l$ queries. Let $\ell = 1$ and $l_{\mathrm{total}} = 0$. In the first round run FPGS$_{\delta, \lceil \ell \rceil}$, add $\lceil \ell \rceil$ to $l_{\mathrm{total}}$, and if a marked state is found, then stop. If after each round, if no marked states have been found, then replace $\ell$ with $\alpha \ell$, run FPGS$_{\delta, \lceil \ell \rceil}$, add $\lceil \ell \rceil$ to $l_{\mathrm{total}}$, and if a marked state is found, then stop. Again, let $\tau_{\delta, \alpha}$ be the $\sqrt{\lambda}$ time the expected value of $l_{\mathrm{total}}$. We provide evidence for the following conjecture.

\begin{conjecture}
    \label{conjecture:tau}
    Let $\delta = 0.4038$ and $\alpha = 1.975$. Then $\tau \leqslant 1.434$.
\end{conjecture}

\begin{proof}[Numerical evidence for \Cref{conjecture:tau}:]
    We can write
    \begin{align}
        \tau_{\delta, \alpha}   &= \sqrt{\lambda} \mathbb{E} \left[ l_{\mathrm{total}} \right] \\
                                &= \sqrt{\lambda} \sum\limits_{s = 1}^\infty \( \sum\limits_{r = 1}^s \lceil \alpha^{r - 1} \rceil \) \mathbb{P} \( \mbox{FPGS} _{\delta, \lceil \alpha^{s - 1} \rceil} \mbox{ succeeds first} \) \\
                                &= \sqrt{\lambda} \sum\limits_{s = 1}^\infty \lceil \alpha^{s - 1} \rceil \sum\limits_{r = s}^\infty \mathbb{P} \( \mbox{FPGS} _{\delta, \lceil \alpha^{r - 1} \rceil} \mbox{ succeeds first} \) \\
                                &= \sqrt{\lambda} \sum\limits_{s = 1}^\infty \lceil \alpha^{s - 1} \rceil \displaystyle\prod_{r = 1}^{s - 1} \mathbb{P} \( \mbox{FPGS} _{\delta, \lceil \alpha^{r - 1} \rceil} \mbox{ fails} \).
    \end{align}
    Let $l_s \eqdef \lceil \alpha^{s - 1} \rceil$. Then by \cref{eq:Psuccess} we have
    \begin{equation}
        \mathbb{P} \( \mbox{FPGS} _{\delta, l_s} \mbox{ fails} \) = \( T_{2 l_s + 1} \( \sqrt{1 - \lambda} T_{\nicefrac{1}{2 l_s + 1}} \( \nicefrac{1}{\delta} \) \) \)^2 \delta^2.
    \end{equation}
    Let
    \begin{align}
        T_{s, \delta}   &\eqdef \( 2 l_s + 1 \) \( \sqrt{1 - \lambda} T_{\nicefrac{1}{2 l_s + 1}} \( \nicefrac{1}{\delta} \) \) \\
                        &= \( 2 l_s + 1 \) \( \sqrt{1 - \lambda} \cosh \( \tfrac{\mathrm{arccosh} \( \nicefrac{1}{\delta} \)}{2 l_s + 1} \) \).
    \end{align}    
    When $l_s < l_{\mathrm{crit}}$, then this yields
    \begin{equation}
        \mathbb{P} \( \mbox{FPGS} _{\delta, l_s} \mbox{ fails} \) = \cosh^2 \( \mathrm{arccosh} \( T_{s, \delta} \) \)^2 \delta^2,
    \end{equation}
    otherwise
    \begin{align}
        \mathbb{P} \( \mbox{FPGS} _{\delta, l_s} \mbox{ fails} \)   &\leqslant \mathbb{P} \( \mbox{FPGS} _{\delta, l_{s - 1}} \mbox{ fails} \) \delta^2 \\
                                                                    &\leqslant \mathbb{P} \( \mbox{FPGS} _{\delta, l_{s_0}} \mbox{ fails} \) \delta^{2 (s - s_0)},
    \end{align}
    where $s_0$ is the smallest integer such that $l_{s_0} \geqslant l_{\mathrm{crit}}$. Let
    \begin{equation}
        Q_s \eqdef \displaystyle\prod_{r = 1}^{s - 1} \mathbb{P} \( \mbox{FPGS} _{\delta, \lceil \alpha^{r - 1} \rceil} \mbox{ fails} \).
    \end{equation}
    Thus we get, as long as $\alpha \delta^2 < 1$, that
    \begin{align}
        \tau_{\delta, \alpha}   &= \sqrt{\lambda} \sum\limits_{s = 1}^{s_0 - 1} \lceil \alpha^{s - 1} \rceil \displaystyle\prod_{r = 1}^{s - 1} \mathbb{P} \( \mbox{FPGS} _{\delta, \lceil \alpha^{r - 1} \rceil} \mbox{ fails} \) + \\
                                & \quad + \sqrt{\lambda} \sum\limits_{s = s_0}^\infty \lceil \alpha^{s - 1} \rceil \displaystyle\prod_{r = 1}^{s - 1} \mathbb{P} \( \mbox{FPGS} _{\delta, \lceil \alpha^{r - 1} \rceil} \mbox{ fails} \) \\
                                &\leqslant \overbrace{\sqrt{\lambda} \sum\limits_{s = 1}^{s_0 - 1} \lceil \alpha^{s - 1} \rceil \displaystyle\prod_{r = 1}^{s - 1} \( \cosh^2 \( \mathrm{arccosh} \( T_{s, \delta} \) \)^2 \delta^2 \)}^{\mathbb{E}_{\lambda, \delta, \alpha} \eqdef} + \\
                                & \quad + \sqrt{\lambda} Q_{s_0} \alpha^{s_0 - 1} \sum\limits_{s = s_0}^\infty \alpha^{s - s_0} \delta^{2 (s - s_0)} \\
                                &= \mathbb{E}_{\lambda, \delta, \alpha} + \tfrac{\sqrt{\lambda} Q_{s_0}}{1 - \alpha \delta^2}.
    \end{align}
    In the range where $\alpha \delta^2 < 1$, $\mathbb{E}_{\lambda, \delta, \alpha}$ is bounded, but it is hard to analytically prove the convergence of $\tau_{\delta, \alpha}$ as $\lambda \rightarrow 0^+$. Nevertheless, numerical simulations suggest that for $\delta = 0.4038$ and $\alpha = 1.975$, $\tau_{\delta, \alpha}$ converges to $\approx 1.433$. The results of a numerical evaluation of the above upper bound for a wide range of parameters is shown in \Cref{figure:phase_portrait}.
    \begin{figure}[h]
        \centering
        \includegraphics[scale=.5]{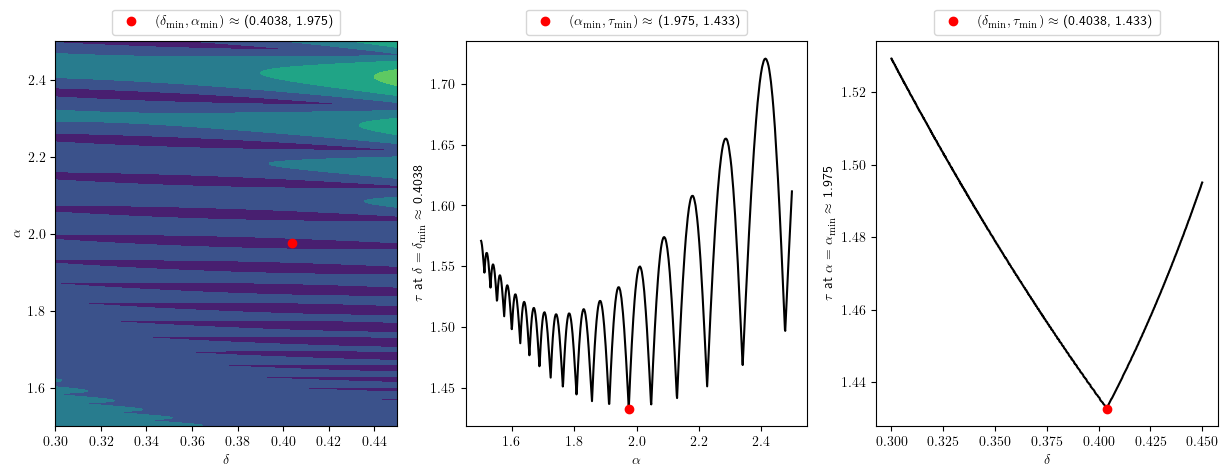}
        \caption{The phase portrait of $\tau_{\delta, \alpha}$ at $\lambda = 2^{- 40}$.}
        \label{figure:phase_portrait}
    \end{figure}
\end{proof}

\smallskip

\begin{remark}
    Note that while the value of the optimal $\delta$ appears sharp on \Cref{figure:phase_portrait}, the graph of the $\alpha$-$\tau$ dependence has many local minima and, furthermore, those minima correspond to lower $\alpha$'s but similar $\tau$'s. Thus, when the depth of the circuits is a concern, which will probably be the case for early fault-tolerant quantum computers, then such parameter values might be more beneficial.
\end{remark}

\smallskip

Python code implementing the above computation can be found in the repository complementing this paper: \cite{github_fpgs_for_qubo}.

\bigskip

\subsection{Time complexity}
\label{sec:time}

In this section, we analyze the total time complexity of the FPGS for QUBO, using our oracle design and assuming \Cref{conjecture:tau}.

\smallskip

The expected time complexity of the Fixed-point Grover Adaptive Search in \Cref{sec:parameter_optimization_for_FPGAS} depends on three quantities:
\begin{enumerate}

    \item The ratio $\lambda$ is marked states to all states.
     
    \item The circuit depth of the diffusion operator, $S_s \( \alpha \)$. Let us call this complexity $C_s$.
     
    \item The circuit depth of the operator, $S_t \( \beta \)$. Let us call this complexity $C_t$.

\end{enumerate}
The time complexity, assuming \Cref{conjecture:tau}, then is $\approx \tfrac{1.433}{\sqrt{\lambda}} \( C_s + C_t \)$.

Now let us assume that a QUBO problem is given with $S_s \( \alpha \)$ and $S_t \( \beta \)$ as defined in \Cref{sec:grover_for_qubo}.

The complexity of the diffusion operator is $C_s = O \( \log_2 (n) \)$.

The complexity, $C_t$, depends on the number of ancillas used. In the case of the ancilla-free oracle, straightforward computation shows that the depth of $U_{f, d}$ is $O  \( \( m + d \) \log_2 \( d \) \)$, where $m = O \( n^2 \)$ is the number of nonzero terms in \cref{eq:f_rewrite}. Since in our applications the ancillas start and end in the all zero state, the inital QFT and the terminal Inverse QFT can be replaced by $d$ many parallel Hadamard gates (the ``Quantum Walsh--Hadamard Transform'') which is much shallower and easier-to-implement. However the QFT and its inverse in the middle of $U_t$ need to be implemented, which can be done with approximately with a $T$ count of $O \( d \log_2 \( d \) \)$, and $T$-depth of $O \( d \)$; cf. \cite{nam_approximate_2020}. Beyond these two subroutines, the oracle uses only CNOT gates and $R_Z$ rotation. In fault-tolerant implementations further $T$ gates come the latter. The absolute values of the angles of the $R_Z$ gates are at least $\tfrac{\pi q_f}{2^{d - 1}}$, where $q_f$ is the smallest of the absolutes value of the nonzero coefficients appearing in the $\cP_d$ gates above. Thus, fault-tolerant implementations of such gates require a $T$-depth (and count) of $O \( d + \log_2 \( \tfrac{1}{q_f} \) \)$; cf. \cite{gheorghiu_t_2022}. Finally, when ancillas are allowed, the non-Clifford cost is identical, the Clifford cost only grows by an $O(1)$ factor, and the depth is given in \Cref{sec:ancillas}. Thus, $C_t$ ranges from $O \( \log_2 \( n d \) \)$ to $O \( \log_2 \( d \) n^2 \)$.

Put the above together, we get that the best total time complexity (as long as $d = O (n)$) is $O \( \tfrac{\log_2 (n)}{\sqrt{\lambda}} \)$. Note that this beats random guessing, which is $O \( \tfrac{1}{\lambda} \)$, as long as $\lambda \ll \tfrac{1}{\log_2 (n)^2}$.

\smallskip

\begin{remark}
    \label{remark:graphcuts}
    In the case when $f$ is the cut function of a simple, unoriented, and undirected graph, then $n$ is the number vertices, $m$ is the number of edges, $Q$ is the graph Laplacian, and $d$ can be chosen to be $O \( \log_2 (n) \)$. Thus, in the shallowest case, one query in the FPGS for maximum graph cuts can be implemented on $n + O \( m \log (n) \)$ qubits, with a depth of $O \( \log_2 (n) \)$.
     
    The Edward--Erd\H{o}s bound for maximum cuts is $\tfrac{m}{2} + \tfrac{n - 1}{4}$ and the best known classical method \cite{crowston_max_2015} for finding such cuts have time complexity $O \( n^4 \)$. It is an interesting and nontrivial question whether FPGS could yield a faster algorithm. In particular, if the ratio of cuts above the Edward--Erd\H{o}s bound among all, $2^n$ many, cuts is $\lambda \in \( 0, 1 \)$, then for graphs that satisfy
    \begin{equation}
        \tfrac{\log_2 (n)^2}{n^8} = o \( \lambda \),
    \end{equation}
    the FPGS with our marker oracle design outperforms the classical method in \cite{crowston_max_2015}. Further investigation of this question is however beyond the scope of this paper.
\end{remark}

\bigskip

\section{Fixed-point Grover Adaptive Search for QUBO}
\label{sec:adaptive}

In this section, we define our classical-quantum method for QUBO problems, where the classical part is an adaptive method that assigns the parameters to the quantum search. This is done by building on ideas of \cites{gilliam_grover_2021,li_quantum_2019}, with further improvements coming from \Cref{sec:grover_for_qubo}. The full method is summarized in \Cref{figure:flowchart}.

\smallskip

Let us assume that we are given an $n$-dimensional instance of a QUBO problem, $f$, so that its values are representable on $d = O (n)$ bit numbers. Assume also that we have a ``stopping condition'' described by a function
\begin{equation}
    \rl_+ \times \Z \ni \( t, th \) \mapsto \textsc{stop}_f \( t, th \) \in \left\{ \mathrm{yes}, \mathrm{no} \right\},
\end{equation}
where $t$ is the elapsed time and $th$ is the highest value found at time $t$. Our goal is to find the highest value parameter before the stopping condition triggers.

Let $\mathrm{G} \( th, l \)$ be the FPGS for QUBO circuit with $\delta = 0.4038$, threshold $th$, and $l$ queries. With the above arguments in mind, we propose an adaptive, quantum--classical hybrid method, which we call the \textit{Fixed-point Grover Adaptive Search for QUBO}, as follows:
\begin{itemize}
    \item[\textbf{Start}] Choose $x \in \F_2^n$ at random and set the threshold to $th = f \( x \)$. Set $l = 1$.
    \item[\textbf{Search}] While the stopping condition is not satisfied, search for large values of $f$ using FPGS for QUBO with $\delta = 0.4038$. After each search, update $l$ to $1.975 \cdot l$ and, if a large value is found, update $x$ to the corresponding configuration.
    \item[\textbf{Stop}] When the stopping condition is satisfied, output the configuration corresponding to the largest observed value.
\end{itemize}
The flowchart of the Fixed-point Grover Adaptive Search for QUBO is shown in \Cref{figure:flowchart} below.
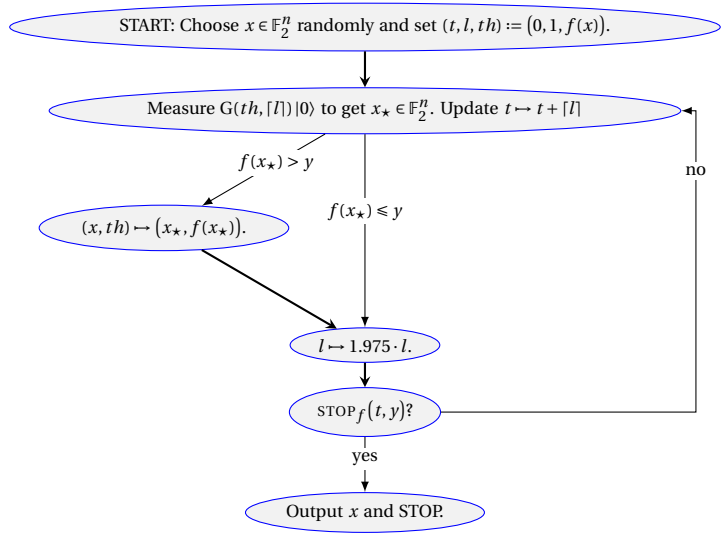
\begin{figure}[!ht]
    \centering
    \hfill\\
    {\scriptsize \begin{tikzpicture}[node distance=2em, scale=0.8, every node/.style={transform shape}]

    \node[process] (start) {START: Choose $x \in \F_2^n$ randomly and set $\( t, l, th \) \eqdef \( 0, 1, f \( x \) \)$.};
    \node[process, below of=start, yshift= -3em] (FPGS) {Measure $\mathrm{G} \( th, \lceil l \rceil \) \ket{0}$ to get $x_\star \in \F_2^n$. Update $t \mapsto t + \lceil l \rceil$};
    \node[process, below of=FPGS, xshift=-12em, yshift=-5em] (x_update_yes) {$\( x, th \) \mapsto \( x_\star, f \( x_\star \) \)$.};
    \node[process, below of=x_update_yes, xshift=12em, yshift=-5em] (l_update) {$l \mapsto 1.975 \cdot l$.};
    \node[process, below of=l_update, yshift=-2em] (fork) {$\textsc{stop}_f \( t, y \)$?};
    \node[process, below of=fork, yshift=-4em] (end) {Output $x$ and STOP.};
    
    \draw[arrow] (start) -- (FPGS);
    \draw[-latex] (FPGS) -- node[pos=0.4,fill=white,inner sep=2pt]{$f \( x_\star \) > y$} (x_update_yes);
    \draw[-latex] (FPGS) -- node[pos=0.4,fill=white,inner sep=2pt]{$f \( x_\star \) \leqslant y$} (l_update);
    \draw[arrow] (x_update_yes) -- (l_update);
    \draw[arrow] (l_update) -- (fork);
    \draw[-latex] (fork) -- node[pos=0.4,fill=white,inner sep=2pt]{yes} (end);
    \draw[-latex] (fork) -| ++(5.5, 0) |- node[pos=0.4,fill=white,inner sep=2pt]{no} (FPGS);

    \end{tikzpicture} }
    \caption{Fixed-point Grover Adaptive Search for QUBO}
    \label{figure:flowchart}
\end{figure}

\medskip

The above algorithm has a further benefit that is not present in the original, one-time search case: as the threshold value increases, the set of marked states decreases and hence the query complexity with which the new highest value configuration is found is a suitable candidate for the next trial-and-error round. We summarize this point in the next lemma. Before that, let us make the following definition: Fix a QUBO problem $f : \F_2^n \rightarrow \Z$. Recall that of $y \in \Z$
\begin{equation}
    \cM_{f, y} = \left\{ x \in \F_2^n | y > f(x) \right\}.
\end{equation}
Now, for each $\epsilon \in \( 0, 1 \)$, let
\begin{equation}
    \cF_{f, \epsilon} \eqdef \left\{ x \middle| \( 1 - \epsilon \) 2^n \leqslant \left| \cM_{f, f(x)} \right| \right\},
\end{equation}
that is, $x \in \cF_{f, \epsilon}$ exactly when $x$ is in the top $\epsilon 2^n$ configurations (with respect to the values of $f$). Now we are ready to state and prove our lemma.

\begin{lemma}
    \label{sec:performance_guarantee}
    For each $f$ and $\epsilon$ as above, the Fixed-point Grover Adaptive Search for QUBO finds a configuration in $\cF_{f, \epsilon}$ in $O \( \tfrac{\log_2 (n)}{\sqrt{\epsilon}} \)$ time.
\end{lemma}

\begin{proof}[Heuristic proof of \Cref{sec:performance_guarantee}:]
    Let $g_0 = 1.433$.
    
    If, in the Fixed-point Grover Adaptive Search for QUBO in \Cref{figure:flowchart}, a new configuration is found using a FPGS search with $l$ queries, then the expected number of queries needed for the next successful search is given by
    \begin{equation}
            \begin{split}
            \mathbbm{E} \left[ l_{\mathrm{next}} \middle| l_{\mathrm{previous}} = l \right]& \approx \tfrac{g_0}{\sqrt{\lambda_{\mathrm{next}}}} - l \\ & \approx g_0 \( \tfrac{1}{\sqrt{\lambda_{\mathrm{next}}}} - \tfrac{1}{\sqrt{\lambda_{\mathrm{previous}}}} \)\\ & \approx g_0 \tfrac{\rd F}{2 \( 1 - F \( y_\lambda \) \)^{3/2}},
        \end{split}
    \end{equation}
    where
    \begin{itemize}
        \item $\lambda \eqdef \lambda_{\mathrm{previous}}$,
        \item $F : \rl \rightarrow \left[ 0, 1 \right]$ is the cumulative distribution for $f$, that is,
        \begin{equation}
            F \( y \) \eqdef \tfrac{\left| \cM_{f, y} \right|}{2^n},
        \end{equation}
        \item $y_\lambda \eqdef F^{- 1} \( 1 - \lambda \)$,
        \item and $\rd F \eqdef \lambda_{\mathrm{previous}} - \lambda_{\mathrm{next}}$.
    \end{itemize}
    We assume that $F$ is well-approximated by a smooth distribution.

    The above yields that
    \begin{equation}
        \mathbbm{E} \left[ l_{\mathrm{total}, \epsilon} \right] \approx g_0 \tfrac{1}{\sqrt{1 - F}} \big|_0^{1 - \epsilon} = g_0 \( \tfrac{1}{\sqrt{\epsilon}} - 1 \).
    \end{equation}
    Since the shallowest implementation of a single query for our FPGS for QUBO is $O \( \log_2 (n) \)$, we get that the expected amount of time needed is $O \( \tfrac{\log_2 (n)}{\sqrt{\epsilon}} \)$.
\end{proof}

\medskip

\begin{remark}
    As remarked earlier, our FPGS for QUBO is only expected to outperform random guessing when $\lambda = o \( \tfrac{1}{\log_2 (n)^2} \)$, thus one can potentially further improve the above method by initially doing more rounds of random sampling (up to $\sim \log_2 (n)^2$), using the largest observed value as the first threshold and starting with $l = O \( \log_2 (n) \)$ queries.
\end{remark}

\smallskip

\begin{remark}
    The above results help cut the total resource cost and total time of the algorithm by a constant factor, but, of course, does not change the time complexity in the big-$O$ sense. Note however, that compared to the Grover Adaptive Search for QUBO method of Gilliam et al., we have improved the non-Clifford complexity of each query by (at least) a factor of two and improved the expected number of queries by a factor of $\approx \tfrac{2.25}{1.433} \approx 3.14$, yielding a total improvement of $\approx 6.28$. Similarly, both of these results yield reductions in time the overall time complexity of the method, especially when ancillas are used.
\end{remark}

\bigskip

\section{Comparison to Grover Adaptive Search}
\label{sec:comparison}

In this section, we give a detailed comparison of the Grover Adaptive Search of \cite{gilliam_grover_2021} and our Fixed-point method. We first contrast the oracle designs and we construct a simulation to compare performances up to $n \leqslant n$ bit QUBOs.

\smallskip

The authors in \cite{gilliam_grover_2021} gave a design similar to our $U_{f, d}$ without ancillas. However, their construction uses simply-controlled (for diagonal terms) and doubly-controlled (for off-diagonal term) phase gates, which are more costly than ours using fanout gates and uncontrolled phase gates. Furthermore, in their design each control qubit, $\ket{x_j}$, needs to access each target qubit, $\ket{y_a}$, sequentially. Thus, even assuming a charitable interpretation of their design (that is, assuming that instead of controlled phase gates, controlled $R_Z$ gates are used) and a good compilation of the circuit, our designs yields an at least a $2$-fold improvement in the number of non-Clifford gates (as a controlled $R_Z$ gate costs two single qubit $R_Z$ gates), an at least $\approx d$-fold reduction in $R_Z$ depth (as now $d$ of them can be run parallel), and an at least $\approx \tfrac{d}{\log_2 (d)}$-fold reduction in total depth.

Finally, our design can make use of ancilla qubits to achieve further depth reduction without any extra non-Clifford cost; a feature not present in \cite{gilliam_grover_2021}.

\smallskip

Next we benchmark the performances of these methods on a concrete QUBO problem, by rewriting the evolution of the distributions as time-dependent discrete-time Markov chains. We describe the details of the simulation below.

Let us assume that an instance of QUBO is given, called $f$, let $F : \rl \rightarrow \left[ 0, 1 \right]$ be the cumulative distribution of $f$ as above, and let $y \in \F_2^n$. When one runs FPGS for a threshold $f(y)$, with parameter $\delta$ and $l$ queries, then the probability of success if given by
\begin{equation}
    \P_{\scriptscriptstyle \mathrm{FPGS}, y, \delta, l} = 1 - \( T_{2 l + 1} \( T_{\nicefrac{1}{2 l + 1}} \( \nicefrac{1}{\delta} \) \sqrt{1 - F(y)} \) \)^2 \delta^2.
\end{equation}
When $x \in \F_2^n$ with $f(x) > f(y)$ is found, then the threshold changes in our Fixed-point Grover Adaptive Search for QUBO and the role of $y$ (as the ``best configuration seen'') is now taken by $x$, otherwise $y$ keeps its role. Thus, the probability of a given $x$ replacing $y$ is
\begin{equation}
    \P_{\scriptscriptstyle \mathrm{FPGS}} \( y \rightarrow x, \delta, l \) = \tfrac{\P_{\scriptscriptstyle \mathrm{FPGS}, y, \delta, l}}{\( 1 - F(y) \) 2^n},
\end{equation}
where the denominator $\( 1 - F(y) \) 2^n$ is the number of configurations $x \in \F_2^n$ with $f(x) > f(y)$. The probability of $y$ keeping its role is
\begin{equation}
    \P_{\scriptscriptstyle \mathrm{FPGS}} \( y \rightarrow y, \delta, l \) = 1 - \P_{\scriptscriptstyle \mathrm{FPGS}, y, \delta, l}.
\end{equation}
All other transitions are impossible. Now assume that the $y$ was randomly picked by a distribution $p : \F_2^n \rightarrow \left[ 0, 1 \right]$. Then the probability of changing having $x$ after the above FPGS is given by
\begin{equation}
    \P_{\scriptscriptstyle \mathrm{FPGS}} (x | p, \delta, l) = \sum\limits_{y \in \F_2^n} p(y) \P_{\scriptscriptstyle \mathrm{FPGS}} \( y \rightarrow x, \delta, l \).
\end{equation}
Similarly, using \cite{boyer_tight_1998}*{Lemma~2}, one sees that in the case of (non-fixed-point) Grover Adaptive Search, the analogous formulae are given by
\begin{align}
    \theta_y                                                        &= \tfrac{\mathrm{arcsin} \( \sqrt{1 - F(y)} \)}{\pi}, \\
    \P_{\scriptscriptstyle \mathrm{GAS}, y, m}                      &= \tfrac{1}{2} \( 1 - \tfrac{\mathrm{sinc} \( 4 m \theta \)}{\mathrm{sinc} \( 2 \theta \)} \), \\
    \P_{\scriptscriptstyle \mathrm{GAS}} \( y \rightarrow x, m \)   &= \left\{ \begin{array}{ll} \tfrac{\P_{\scriptscriptstyle \mathrm{GAS}, y, m}}{\( 1 - F(y) \) 2^n}, & \mbox{if } f(x) > f(y) \\ 1 - \P_{\scriptscriptstyle \mathrm{GAS}, y, m}, & \mbox{if } x = y, \\ 0, & \mbox{otherwise.} \end{array} \right. \\
    \P_{\scriptscriptstyle \mathrm{GAS}} (x | p, m)                 &= \sum\limits_{y \in \F_2^n} p(y) \P_{\scriptscriptstyle \mathrm{GAS}} \( y \rightarrow x, m \).
\end{align}

In the first step of both our Fixed-point Grover Adaptive Search for QUBO and the Grover Adaptic Search for QUBO of Gilliam et al., one choses $y \in \F_2^n$ at random, that is $p_0 (y) = \tfrac{1}{2^n}$ and then runs the appropriate version with $l = m = 1$, updates $y$ if a larger value was found, and finally updates $l$ and $m$. Given $f$, the above formulae allow us to compute the probabilities of getting configurations in either of these methods.

\smallskip

Let now $f$ be cut function of the randomly generated (connected) Erd\H{o}s--R\'enyi graph with $n = 32$ vertices shown in \Cref{figure:graph}. The graph has a maximum cut of $146$ and the expected value of a random cut is $114.5$, which is $78.42$\% of the maximum cut. The variance is $5.18$\%.

\begin{figure}[h!]
    \centering
    \includegraphics[scale=.7]{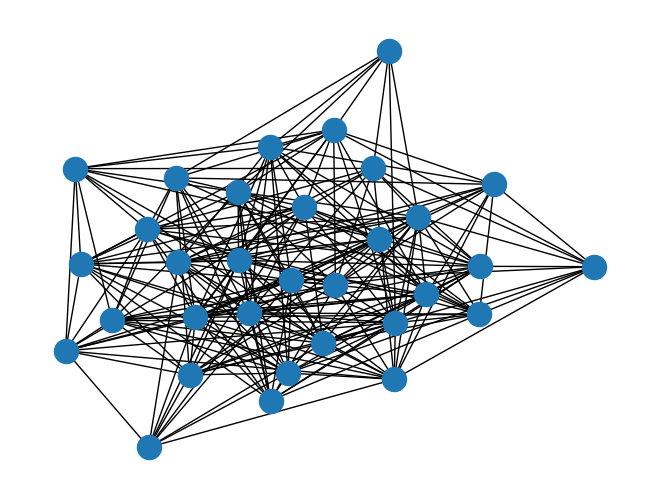}
    \caption{The test graph.}
    \label{figure:graph}
\end{figure}

We ran both our Fixed-point Grover Adaptive Search for QUBO and the Grover Adaptic Search for QUBO of Gilliam et al. for four iterations. Straightforward computations, using the results of the previous sections, yield that the resource requirements of the two methods are similar in this case, so that the comparison is meaningful. We found that the expected value of cut that our method produces is $88.19$\% of the maximum cut, compared to the $84.86$\% of the method of Gilliam et al., and the variances are $2.6$\% and $3.33$\%, respectively. Our method amplified the chance of finding the maximum cut by a factor of more than $88$, compared to random guesses, and by a factor of $8$, compared to the Grover Adaptic Search for QUBO of Gilliam et al. The cut-distributions of a randomly picked configuration (``before Grover''), the distribution after the Grover Adaptic Search for QUBO of Gilliam et al., and after our Fixed-point Grover Adaptive Search for QUBO is given on \Cref{figure:cuts}.

\begin{figure}[!h]
    \centering
    \includegraphics[scale=.5]{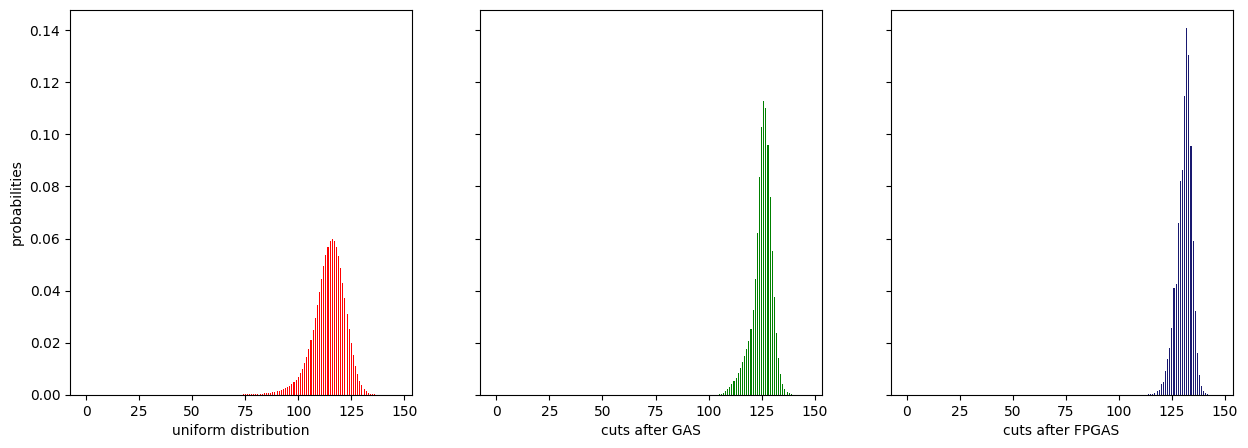}
    \caption{The cut-distributions.}
    \label{figure:cuts}
\end{figure}

\medskip

\begin{remark}
    While in the simulations we used the above $\alpha$ parameter for our Fixed-point Grover Adaptive Search for QUBO and $\Lambda = \tfrac{6}{5}$ for the Grover Adaptic Search for QUBO of Gilliam et al., which is recommended in \cite{boyer_tight_1998}*{Section~4}, we saw that for small bit numbers ($n \lesssim 32$), different values might work better, albeit the change is marginal. The code to reproduce our testing is available in \cite{github_fpgs_for_qubo}.
\end{remark}

\smallskip

\begin{remark}
    Word of caution: For the above mentioned small bit numbers, repeated random guess $O \( n^2 \)$ times (so at the same time complexity of our method) performs comparably.
\end{remark}

\bigskip

\section{Conclusion}
\label{sec:conclusion}

We constructed a marker oracle for QUBO problems with improved complexities to the previously known designs. These oracles are expected to be useful for both Grover-, and QAOA-type quantum circuits. We also studied adaptive optimization methods, using our constructions in conjunction with FPGS. By the results of \cite{li_quantum_2019} and our improved oracle design, it is immediate that our Fixed-point Grover Adaptive Search for QUBO has better performance guarantees than that of the original Grover Adaptive Search of \cite{gilliam_grover_2021}. However, more extensive studies of the theoretical foundations (such as a better understanding of the distribution of values in QUBO problems) as well as further circuit simplification and resource estimates are needed.

\medskip

\subsection*{Code and Data Availability:} Supplementary Qiskit code for this paper is available in \cite{github_fpgs_for_qubo}.

\smallskip

\subsection*{Acknowledgments:} We are grateful to Amazon Web Services and QLab for providing credits to access IonQ's QPUs. We also thank Constantin Gonciulea, Tom Ginsberg, Jan Tu\l{}owiecki, and Shahaf Asban for providing valuable feedback. Finally, we thank QuForce for bringing the authors together for the project.

\bigskip

\appendix

\section{Experimental testing of the encoder oracle testing on IonQ's Quantum Computers:} The above marker oracle was tested on IonQ's Aria $2$ QPU. We used $9$ qubits ($5$-bit QUBO with $4$ ancillas) with $5000$ shots. The matrix of the quadratic form was
\begin{equation}
    Q = \begin{pmatrix}
            2 & - 1 & 0 & - 1 & 0 \\
            - 1 & 1 & 0 & 0 & 0 \\
            0 & 0 & 2 & 0 & - 1 \\
            - 1 & 0 & 0 & 2 & 0 \\
            0 & 0 & - 1 & 0 & 2
        \end{pmatrix},
\end{equation}
with threshold $y = 4$. Since the maximum of the integer-valued, quadratic polynomial $f \( x \) = x^T Q x$ is $5$, only maximal configurations should be marked, that is, $\ket{01110},  \ket{01011}$, and $\ket{01111}$. The oracle was decomposed into $1$-qubit and CNOT gates, yielding a circuit depth of $47$ and a CNOT count of $66$. The experimental results are given in \Cref{table:oracle} below.

\begin{table}[!ht]
    \centering
    \begin{tabular}{|c|c|c|c|c|}
        \hline
        $x$ & $f (x)$ & $y_0 = 0$ & $y_0 = 1$ & \textnormal{total} \\
        \hline
        00000 & 0 & {\color{ForestGreen} 3.44\%} & {\color{red} 0.74\%} & 4.18\% \\
        10000 & 2 & {\color{ForestGreen} 1.68\%} & {\color{Red} 0.36\%} & 2.04\% \\
        10000 & 1 & {\color{ForestGreen} 3.18\%} & {\color{Red} 0.52\%} & 3.70\% \\
        10000 & 1 & {\color{ForestGreen} 4.08\%} & {\color{Red} 0.86\%} & 4.94\% \\
        10000 & 2 & {\color{ForestGreen} 2.16\%} & {\color{Red} 0.68\%} & 2.84\% \\
        10100 & 4 & {\color{ForestGreen} 2.02\%} & {\color{Red} 0.40\%} & 2.42\% \\
        01100 & 3 & {\color{ForestGreen} 1.34\%} & {\color{Red} 0.36\%} & 1.70\% \\
        11100 & 3 & {\color{ForestGreen} 2.90\%} & {\color{Red} 0.32\%} & 3.22\% \\
        00010 & 2 & {\color{ForestGreen} 1.34\%} & {\color{Red} 0.36\%} & 1.70\% \\
        10010 & 2 & {\color{ForestGreen} 5.26\%} & {\color{Red} 0.74\%} & 6.00\% \\
        01010 & 3 & {\color{ForestGreen} 2.52\%} & {\color{Red} 0.34\%} & 2.86\% \\
        11010 & 1 & {\color{ForestGreen} 2.84\%} & {\color{Red} 0.38\%} & 3.22\% \\
        00110 & 4 & {\color{ForestGreen} 1.58\%} & {\color{Red} 0.20\%} & 1.78\% \\
        10110 & 4 & {\color{ForestGreen} 2.94\%} & {\color{Red} 0.88\%} & 3.82\% \\
        01110 & 5 & {\color{Red} 1.06\%} & {\color{ForestGreen} 1.52\%} & 2.58\% \\
        11110 & 3 & {\color{ForestGreen} 2.20\%} & {\color{Red} 0.36\%} & 2.56\% \\
        \hline
    \end{tabular}

    \begin{tabular}{|c|c|c|c|c|}
        \hline
        $x$ & $f (x)$ & $y_0 = 0$ & $y_0 = 1$ & \textnormal{total} \\
        \hline
        00001 & 2 & {\color{ForestGreen} 1.68\%} & {\color{Red} 0.66\%} & 2.34\% \\
        10001 & 4 & {\color{ForestGreen} 1.06\%} & {\color{Red} 0.26\%} & 1.32\% \\
        10001 & 3 & {\color{ForestGreen} 1.04\%} & {\color{Red} 0.30\%} & 1.34\% \\
        10001 & 3 & {\color{ForestGreen} 2.02\%} & {\color{Red} 0.48\%} & 2.50\% \\
        10001 & 2 & {\color{ForestGreen} 3.90\%} & {\color{Red} 0.78\%} & 4.68\% \\
        10101 & 4 & {\color{ForestGreen} 2.30\%} & {\color{Red} 0.36\%} & 2.66\% \\
        01101 & 3 & {\color{ForestGreen} 3.42\%} & {\color{Red} 0.58\%} & 4.00\% \\
        11101 & 3 & {\color{ForestGreen} 4.54\%} & {\color{Red} 0.34\%} & 4.88\% \\
        00011 & 4 & {\color{ForestGreen} 1.14\%} & {\color{Red} 0.38\%} & 1.52\% \\
        10011 & 4 & {\color{ForestGreen} 2.82\%} & {\color{Red} 0.44\%} & 3.26\% \\
        01011 & 5 & {\color{Red} 0.52\%} & {\color{ForestGreen} 1.16\%} & 1.68\% \\
        11011 & 3 & {\color{ForestGreen} 1.60\%} & {\color{Red} 0.52\%} & 2.12\% \\
        00111 & 4 & {\color{ForestGreen} 1.88\%} & {\color{Red} 0.26\%} & 2.14\% \\
        10111 & 4 & {\color{ForestGreen} 5.94\%} & {\color{Red} 1.74\%} & 7.68\% \\
        01111 & 5 & {\color{Red} 0.94\%} & {\color{ForestGreen} 2.46\%} & 3.40\% \\
        11111 & 3 & {\color{ForestGreen} 4.18\%} & {\color{Red} 0.74\%} & 4.92\% \\
        \hline
    \end{tabular}
    
    \caption{Test results with threshold $4$; values below or equal to $4$ should not be marked ($y_0 = 0$) and above $4$ should be marked ($y_0 = 1$). The {\color{ForestGreen} green} numbers are the percentages of the given state being measured with the correct marking, while the {\color{Red} red} ones are the percentages of incorrect markings. The total occurrence of a state is the sum of the corresponding {\color{ForestGreen} green} and {\color{Red} red} percentages which in the ideal case should all be $\nicefrac{3}{32} = 3.125\%$.}
    \label{table:oracle}
\end{table}

While the oracle is $2$--$5$ times more likely to correctly mark states, further experiments and simulations (not listed here) show that contemporary QPUs are too noisy to run our circuits without error correction. Furthermore, it is worth noting that the distribution of states is no longer uniform.

    \bibliography{references}

\end{document}